\documentclass{article}
\usepackage{setspace, amsmath, amssymb, amsthm}
\usepackage[colorlinks = true]{hyperref}
\usepackage[margin = 2.5cm]{geometry}
\usepackage[english]{babel}
\usepackage[utf8x]{inputenc}
\usepackage{amsmath}
\usepackage{graphicx}
\usepackage[usenames,dvipsnames]{pstricks}
\usepackage{epsfig}
\usepackage{pst-grad} % For gradients
\usepackage{pst-plot} % For axes
\usepackage{epstopdf}
\usepackage{subcaption}
\usepackage{tikz}
\usetikzlibrary{arrows}
\usetikzlibrary{calc}
\usetikzlibrary{fit}
\newtheorem{theorem}{Theorem}[section]
\newtheorem{lemma}[theorem]{Lemma}

\newtheorem{corollary}[theorem]{Corollary}

\usepackage{tikz}
\usepackage{tkz-graph}
\tikzstyle{vertex}=[circle, draw, inner sep=0pt, minimum size=6pt]

\title{$\mathcal{B}$-partitions, application to determinant and permanent of graphs}
\author{Ranveer Singh \thanks{Center for System Science, Indian Institute of Technology Jodhpur. email: \texttt{pg201283008@iitj.ac.in}}\\ R. B. Bapat \thanks{Stat-Math Unit, Indian Statistical Institute Delhi, 7-SJSS Marg, New Delhi - 110 016. email: \texttt{rbb@isid.ac.in}}}

\begin{document}
    %\begin{titlepage}
%        \pagecolor{green!50}
        \maketitle
    %\end{titlepage}
%
%    \pagecolor{white}
%    \doublespacing

\begin{abstract}
Let $G$ be a graph(directed or undirected) having $k$ number of blocks. A $\mathcal{B}$-partition of $G$ is a partition into $k$ vertex-disjoint subgraph $(\hat{B_1},\hat{B_1},\hdots,\hat{B_k})$ such that $\hat{B}_i$ is induced subgraph of $B_i$ for $i=1,2,\hdots,k.$ The terms $\prod_{i=1}^{k}\det(\hat{B}_i),\ \prod_{i=1}^{k}\text{per}(\hat{B}_i)$ are det-summands and per-summands, respectively, corresponding to the $\mathcal{B}$-partition. The determinant and permanent of a graph having no loops on its cut-vertices is equal to summation of det-summands and per-summands, respectively, corresponding to all possible $\mathcal{B}$-partitions. Thus, in this paper we calculate determinant and permanent of some graphs, which include block graph with negatives cliques, signed unicyclic graph, mix complete graph, negative mix complete graph, and star mix block graphs.     

\end{abstract}

\section{Introduction}
A simple graph $G$ consists of a finite set of vertices $V(G)$ and set of edges $E(G)$ consisting of distinct, unordered pairs of vertices. Thus, $(i,j)$ or $(j,i)$ represent an edge between vertices $i,j \in V(G)$, and $i$, $j$ are called adjacent vertices. If $E(G)$ consists of ordered pairs of vertices then $G$ is called a directed graph or digraph. In this paper most of the study is on simple graphs, thus we use the term 'graph' for simple graphs.   A signed graph is a graph equipped with a weight function $f: E(G) \rightarrow \{-1,0,1\}$. Thus, signed graph may have positive, negative edges with weights $1$, $-1$, respectively. %All the graphs considered in this article are signed graphs. 
Let $G$ be a signed graph on $n$ vertices.   Then, the adjacency matrix $A=(a_{ij})$ of order $n\times n$ associated with $G$ is defined by $$a_{ij}=\begin{cases}
1 & \mbox{if the vertices $i,j$ are connected with a positive edge}\\
-1& \mbox{if the vertices $i,j$ are connected with a negative edge}\\
0 & \mbox{if the vertices $i,j$ are not connected}
\end{cases} $$ where, $1\leq i,j\leq n.$ Signed graph $G$ has a underlying graph $|G|$, in which all the negative edges are replaced by positive edges. Corresponding adjacency matrix of $|G|$ is denoted by $|A|.$ By determinant and permanent of a graph we mean determinant and permanent of its adjacency matrix.

A signed complete graph is a signed graph where each distinct pair of vertices is connected by a positive or negative edge. A signed clique in signed graph $G$ is an induced subgraph which is a signed complete graph.   When each edge of a clique is negative we call it a negative clique. Similarly, if each edge of a clique is positive then, we call it a positive clique. We denote a complete graph on $n$ vertices, having each edge positive,  by $K_n$. A complete graph on $n$ vertices with arbitrary weights on edges is denoted by $wK_n$. By $K^{m,r}_n$, we denote a signed complete graph on $n$ vertices, having a $m$ number of vertex-disjoint negative cliques each having $r$ vertices, and all the other edges are positive except those are in negative cliques \cite{singh2017eigenvalues}.

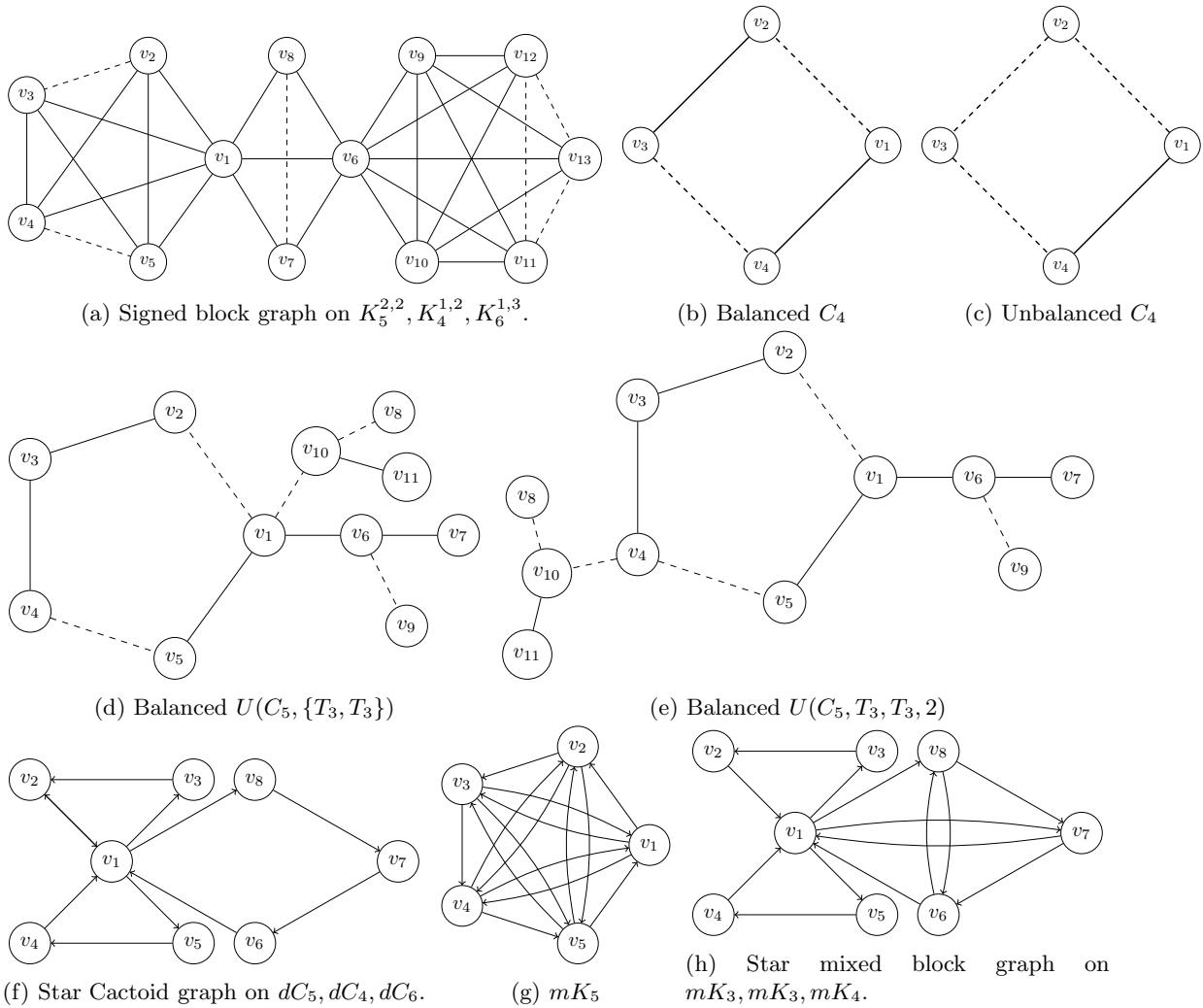
\begin{figure}
        \begin{subfigure}[b]{0.5\textwidth}
        \centering
        \resizebox{\linewidth}{!}{
          \begin{tikzpicture}[scale=1.05]
                       \tikzset{vertex/.style = {shape=circle,draw}}
 \tikzset{edge/.style = {- = latex'}}
                                        % vertices
\node[vertex] (1) at  (2,0) {$v_1$};
  \node[vertex] (2) at (.6180,1.9021) {$v_2$};
\node[vertex] (3) at (-1.6180,1.1756) {$v_3$};
 \node[vertex] (4) at (-1.6180,-1.1756) {$v_4$};
 \node[vertex] (5) at (0.6180,-1.9021) {$v_5$};
 \node[vertex] (6) at (4.3512,0) {$v_6$};
\node[vertex] (7) at (3.17,-1.9) {$v_7$};
\node[vertex] (8) at (3.17,1.9) {$v_8$};
\node[vertex] (9) at (5.57,1.9) {$v_9$};
\node[vertex] (10) at (5.57,-1.9) {$v_{10}$};
\node[vertex] (11) at (7.57,-1.9) {$v_{11}$};
\node[vertex] (12) at (7.57,1.9) {$v_{12}$};
\node[vertex] (13) at (8.57,0) {$v_{13}$};
                 
 \draw[edge] (1) to (2); 
\draw[edge] (1) to (3);  
\draw[edge] (1) to (4); 
\draw[edge] (1) to (5); 
\draw[edge] (2) to (4); 
\draw[edge] (2) to (5);
\draw[edge] (3) to (4);  
\draw[edge] (3) to (5);  
                                        
\draw[edge] (1) to (8); 
\draw[edge] (1) to (7);
\draw[edge] (8) to (6);
\draw[edge] (6) to (7);
\draw[edge] (1) to (6); 
                         
\draw[edge] (6) to (9);
 \draw[edge] (6) to (10); 
\draw[edge] (6) to (11);
                                        \draw[edge] (6) to (12);
                                        \draw[edge] (6) to (13);
                                        \draw[edge] (9) to (10);
                                        \draw[edge] (9) to (11);
                                        \draw[edge] (9) to (12); 
                                        \draw[edge] (9) to (13);
                                        \draw[edge] (10) to (11);
                                        \draw[edge] (10) to (12);
                                        \draw[edge] (10) to (13);

                                   \tikzset{edge/.style = {{dashed}= latex'}} 
                                   \draw[edge] (2) to (3); 
                                   \draw[edge] (4) to (5);
                                   \draw[edge] (7) to (8);
                                   \draw[edge] (11) to (12); 
                                   \draw[edge] (12) to (13);
                                   \draw[edge] (11) to (13);

              \end{tikzpicture}
                }
        \caption{Signed block graph on $K_5^{2,2}, K_4^{1,2}, K_6^{1,3}$.}
        \label{bal}
    \end{subfigure}        
 \begin{subfigure}[b]{0.24\textwidth}
         \centering
         \resizebox{\linewidth}{!}{
             \begin{tikzpicture}[scale=0.6]
                 \SetGraphUnit{2} 
                      %\tikzset{VertexStyle/.style = {shape = circle,fill = black,minimum size = 19pt}}
 \draw  node[draw,circle,scale=1] (1) at (4,0) {$v_1$};
                     
 \draw  node[draw,circle,scale=1](2) at (0,4) {$v_2$};
                      
 \draw  node[draw,circle,scale=1](3) at (-4,0) {$v_3$};
                      
 \draw  node[draw,circle,scale=1](4) at (0,-4) {$v_4$};
                      
 \tikzset{EdgeStyle/.style={dashed}}
                      \Edge(1)(2)
                      \Edge(3)(4)   
                      \tikzset{EdgeStyle/.style={.}}
                      \Edge(2)(3)
                      \Edge(1)(4)      
             \end{tikzpicture}
         }
         \caption{Balanced $C_4$}
         \label{bala}
     \end{subfigure}
     \begin{subfigure}[b]{0.24\textwidth}
     \centering
         \resizebox{\linewidth}{!}{
             \begin{tikzpicture} [scale=0.6]
                \SetGraphUnit{2} 
                     %\tikzset{VertexStyle/.style = {shape = circle,fill = black,minimum size = 19pt}}
                     \draw  node[draw,circle,scale=1] (1) at (4,0) {$v_1$};
                    
                     \draw  node[draw,circle,scale=1](2) at (0,4) {$v_2$};
                     
                     \draw  node[draw,circle,scale=1](3) at (-4,0) {$v_3$};
                     
                     \draw  node[draw,circle,scale=1](4) at (0,-4) {$v_4$};
                     
                     \tikzset{EdgeStyle/.style={dashed}}
                     \Edge(1)(2)
                     \Edge(3)(4)   
                     \Edge(2)(3)
                     \tikzset{EdgeStyle/.style={.}}
                     \Edge(1)(4)     
             \end{tikzpicture}
         }
         \caption{Unbalanced $C_4$}   
         \label{unbal}
     \end{subfigure}

 \begin{subfigure}[b]{0.4\textwidth}
          \centering
          \resizebox{\linewidth}{!}{
              \begin{tikzpicture}[scale=1.1]
                  \SetGraphUnit{2}                       
 \draw  node[draw,circle,scale=1] (1) at (2,0) {$v_1$};
 \draw  node[draw,circle,scale=1] (6) at (3.5,0) {$v_6$}; 
 \draw  node[draw,circle,scale=1] (7) at (5,0) {$v_7$};
 \draw  node[draw,circle,scale=1] (8) at (4,1.9021) {$v_8$};
 \draw  node[draw,circle,scale=1] (10) at (2.8,1.3021) {$v_{10}$};
  \draw  node[draw,circle,scale=1] (11) at (4.2,0.9) {$v_{11}$};
 
 \draw  node[draw,circle,scale=1] (9) at (4.2,-1.4021){$v_9$};
 
 \draw  node[draw,circle,scale=1] (2) at (0.6180,1.9021) {$v_2$};       
 \draw  node[draw,circle,scale=1] (3) at (-1.6180,1.1756) {$v_3$};   
 \draw  node[draw,circle,scale=1] (4) at (-1.6180,-1.1756) {$v_4$};   
 \draw  node[draw,circle,scale=1] (5) at (0.6180,-1.9021) {$v_5$};

  \tikzset{edge/.style = {- = latex'}}
  \draw[edge] (2) to (3);
  \draw[edge] (3) to (4);
  \draw[edge] (5) to (1); 
  \draw[edge] (6) to (1);
  \draw[edge] (6) to (7);
  \draw[edge] (10) to (11); 
  \tikzset{edge/.style = {dashed = latex'}}
 \draw[edge] (1) to (2);
 \draw[edge] (4) to (5);
 \draw[edge] (6) to (9);
 \draw[edge] (1) to (10);
 \draw[edge] (10) to (8);  
\end{tikzpicture}
          }
          \caption{Balanced $U(C_5,\left\{T_{3}, T_{3} \right\})$}
          \label{Uni1}
      \end{subfigure}
 \begin{subfigure}[b]{0.5\textwidth}
          \centering
          \resizebox{\linewidth}{!}{
              \begin{tikzpicture}[scale=1.1]
                  \SetGraphUnit{2}                       
 \draw  node[draw,circle,scale=1] (1) at (2,0) {$v_1$};
 \draw  node[draw,circle,scale=1] (6) at (3.5,0) {$v_6$}; 
 \draw  node[draw,circle,scale=1] (7) at (5,0) {$v_7$};
 \draw  node[draw,circle,scale=1] (8) at (-3.3,-0.30) {$v_8$};
 \draw  node[draw,circle,scale=1] (10) at (-3,-1.4621) {$v_{10}$};
 \draw  node[draw,circle,scale=1] (11) at (-3.3,-2.7) {$v_{11}$};
 
 \draw  node[draw,circle,scale=1] (9) at (4.2,-1.4021){$v_9$};
 
 \draw  node[draw,circle,scale=1] (2) at (0.6180,1.9021) {$v_2$};       
 \draw  node[draw,circle,scale=1] (3) at (-1.6180,1.1756) {$v_3$};   
 \draw  node[draw,circle,scale=1] (4) at (-1.6180,-1.1756) {$v_4$};   
 \draw  node[draw,circle,scale=1] (5) at (0.6180,-1.9021) {$v_5$};

\tikzset{edge/.style = {- = latex'}}
  \draw[edge] (2) to (3);
  \draw[edge] (3) to (4);
  \draw[edge] (5) to (1); 
  \draw[edge] (6) to (1);
  \draw[edge] (6) to (7);
  \draw[edge] (10) to (11); 
  \tikzset{edge/.style = {dashed = latex'}}
 \draw[edge] (1) to (2);
 \draw[edge] (4) to (5);
 \draw[edge] (6) to (9);
 \draw[edge] (4) to (10);
 \draw[edge] (10) to (8);      
\end{tikzpicture}
          }
          \caption{Balanced $U(C_5, T_{3}, T_{3}, 2)$}
          \label{Uni2}
      \end{subfigure}

 \begin{subfigure}[b]{0.35\textwidth}
             \centering
                 \resizebox{\linewidth}{!}{
                     \begin{tikzpicture} [scale=0.7]
                        \SetGraphUnit{2} 
         \draw  node[draw,circle,scale=1] (1) at (0,0) {$v_1$};               
         \draw  node[draw,circle,scale=1] (2) at (-2,2) {$v_2$};
         \draw  node[draw,circle,scale=1] (3) at (2,2) {$v_3$};
         \draw  node[draw,circle,scale=1] (4) at (-2,-2) {$v_4$};
        \draw  node[draw,circle,scale=1] (5) at (2,-2) {$v_5$}; 
        \draw  node[draw,circle,scale=1] (6) at (3.5,-2) {$v_6$}; 
        \draw  node[draw,circle,scale=1] (8) at (3.5,2) {$v_8$};       
        \draw  node[draw,circle,scale=1] (7) at (7,0) {$v_7$};                       
        \tikzset{edge/.style = {-> = latex'}}
        %\draw[edge] (1) to[bend left=10] (3);
        \draw[edge] (1) to (3);
        \draw[edge] (3) to (2);
        \draw[edge] (2) to (1);
        \draw[edge] (1) to (5);
        \draw[edge] (5) to (4);
        \draw[edge] (4) to (1);
        \draw[edge] (8) to (7);
        \draw[edge] (7) to (6);
        \draw[edge] (6) to (1);
        \draw[edge] (1) to (8);
        \draw[edge] (1) to (2); 
   \tikzset{edge/.style = {dashed, -> = latex'}}
             
        %\draw[edge] (1) to[bend left=8] (7);
        %\draw[edge] (7) to[bend left=8] (1);
        %\draw[edge] (8) to[bend left=12] (6);
        %\draw[edge] (6) to[bend left=12] (8); 

                     \end{tikzpicture}
                 }
         \caption{Star Cactoid graph on $dC_5, dC_4, dC_6$.}
         \label{balm}
     \end{subfigure}  
 \begin{subfigure}[b]{0.20\textwidth}
         \centering
         \resizebox{\linewidth}{!}{
             \begin{tikzpicture}[scale=0.9]
                 \SetGraphUnit{2} 
                      %\tikzset{VertexStyle/.style = {shape = circle,fill = black,minimum size = 19pt}}
                      
                      \draw  node[draw,circle,scale=1] (1) at (2,0) {$v_1$};
                      
                \draw  node[draw,circle,scale=1] (2) at (0.6180,1.9021) {$v_2$};

                 \draw  node[draw,circle,scale=1] (3) at (-1.6180,1.1756) {$v_3$};

                  \draw  node[draw,circle,scale=1] (4) at (-1.6180,-1.1756) {$v_4$};   
              \draw  node[draw,circle,scale=1] (5) at (0.6180,-1.9021) {$v_5$};     
 \tikzset{edge/.style = {-> = latex'}}
 \draw[edge] (1) to (2);
 \draw[edge] (2) to (3);
 \draw[edge] (3) to (4);
 \draw[edge] (4) to (5);
 \draw[edge] (5) to (1);                             
                             
\draw[edge] (1) to[bend left=10] (3);
\draw[edge] (3) to[bend left=10] (1);
\draw[edge] (1) to[bend left=10] (4);
\draw[edge] (4) to[bend left=10] (1); 
\draw[edge] (2) to[bend left=10] (4);
\draw[edge] (4) to[bend left=10] (2); 
\draw[edge] (2) to[bend left=10] (5);
\draw[edge] (5) to[bend left=10] (2); 
\draw[edge] (3) to[bend left=10] (5);
\draw[edge] (5) to[bend left=10] (3);

             \end{tikzpicture}
         }
         \caption{$mK_5$}
         \label{balam}
     \end{subfigure}
     \begin{subfigure}[b]{0.35\textwidth}
     \centering
         \resizebox{\linewidth}{!}{
             \begin{tikzpicture} [scale=0.7]
                \SetGraphUnit{2} 
 \draw  node[draw,circle,scale=1] (1) at (0,0) {$v_1$};               
 \draw  node[draw,circle,scale=1] (2) at (-2,2) {$v_2$};
 \draw  node[draw,circle,scale=1] (3) at (2,2) {$v_3$};
 \draw  node[draw,circle,scale=1] (4) at (-2,-2) {$v_4$};
\draw  node[draw,circle,scale=1] (5) at (2,-2) {$v_5$}; 
\draw  node[draw,circle,scale=1] (6) at (3.5,-2) {$v_6$}; 
\draw  node[draw,circle,scale=1] (8) at (3.5,2) {$v_8$};       
\draw  node[draw,circle,scale=1] (7) at (7,0) {$v_7$};                       
\tikzset{edge/.style = {-> = latex'}}
%\draw[edge] (1) to[bend left=10] (3);
\draw[edge] (1) to (3);
\draw[edge] (3) to (2);
\draw[edge] (2) to (1);
\draw[edge] (1) to (5);
\draw[edge] (5) to (4);
\draw[edge] (4) to (1);

\draw[edge] (1) to (8);
\draw[edge] (8) to (7);
\draw[edge] (7) to (6);
\draw[edge] (6) to (1);

\draw[edge] (1) to[bend left=8] (7);
\draw[edge] (7) to[bend left=8] (1);
\draw[edge] (8) to[bend left=12] (6);
\draw[edge] (6) to[bend left=12] (8);

             \end{tikzpicture}
         }
         \caption{Star mixed block graph on $mK_3, mK_3, mK_4$.}   
         \label{unbalm}
     \end{subfigure}

\caption{Examples: Dark line shows positive edge(weight +1), dotted line shows negative edge(weight -1).}

\end{figure}

A path of length $k$ between two vertices $v_1$, and $v_k$ is a sequence of distinct vertices $v_{1}, v_{2}, \hdots, v_{k-1}, v_{k}$, such that, for all $i=1,2, \hdots, k-1$,  $(v_{i}, v_{i +1}) \in E(G)$. If $v_1=v_k$ then, the path is called a cycle. If $G$ is a digraph then we consider a path to be a sequence of distinct vertices $v_{1}, v_{2}, \hdots, v_{k-1}, v_{k}$, such that, for all $i=1,2, \hdots, k-1$,  either $(v_{i}, v_{i +1})$ or $ (v_{i+1}, v_{i}) \in E(G)$.  We call $G$ be connected if there exist a path between any two distinct vertices. A component of $G$ is a maximally connected subgraph of $G$. A cut-vertex of $G$ is a vertex whose removal results increase the number of components in $G$. A block is a maximally connected subgraph of $G$ that has no cut-vertex \cite{bapat2014adjacency}. Note that, if $G$ is a connected graph having no cut-vertex, then $G$ itself is a block. A block having only one cut-vertex of $G$ is called its pendant block. When each block of signed graph $G$ is a complete signed graph then we call it a signed block graph.  An example of such a signed block graph is given in Figure \ref{bal}. We consider  signed block graph in which each block can have vertex disjoint negative cliques having same number of vertices,  and edges connected to cut-vertices are positive,see Figure \ref{bal}.

We also consider weighted signed graphs, that is, edges can have arbitrary weights. Though the meaning of weighted signed graphs and weighted graphs is same, still we will use the word weighted signed graphs just to highlight the importance of signs of edges.  In a weighted signed graph $G$, a cycle is called a balanced cycle if the product of weights on its edges is positive, otherwise, it is called an unbalanced cycle \cite{cartwright1956structural}. In other words, for a cycle to be balanced it should have an even number of edges with negative weights. See figures \ref{bala} and \ref{unbal} for a example of balanced and unbalanced cycle, respectively. A weighted signed graph $G$ is called balanced graph when all the cycles in $G$ are balanced \cite{harary1953notion}. In particular, for complete graph $K_n$ if all the triangles (cycles of length 3) are balanced then it is a balanced graph       \cite{easley2010networks}. Following theorem provide spectral criteria for a weighted signed graph $G$ to be balanced. 

\begin{theorem}\cite{acharya1980spectral}\label{bd}
A weighted signed graph $G$ is balanced if and only if eigenvalues of $G$ and $|G|$ are same.
\end{theorem}

A signed unicyclic graph is a connected signed graph in which number of edges equals the number of vertices. Thus, a signed unicyclic graph is either a cycle or a cycle with trees attached to the vertices of the cycle. If the cycle is balanced then the signed unicyclic graph is balanced otherwise unbalanced. Let $T_m$ denotes a signed tree graph having $m$ vertices. Then, $U(C_n,\{T_{m_1}, T_{m_2},\hdots,T_{m_k}\})$ denotes a signed unicyclic graph having a signed cycle $C_n$ and $k$ signed trees $T_{m_1}, T_{m_2},\hdots,T_{m_k}$ such that root of each $T_{m_i}, i=1,2,\hdots,k$ is linked to a fix vertex of $C_n$. An example of a balanced $U(C_n,\{T_{m_1}, T_{m_2},\hdots,T_{m_k}\})$ is given in Figure \ref{Uni1}. By, $U(C_n,T_{m_1}, T_{m_2},l)$ we denote a unicyclic graph having a signed cycle $C_n$ and roots of trees $T_{m_1}, T_{m_2}$ are attached to two vertices $v_1$ and $v_2$ of $C_n$, respectively, at a distance $l$. An example of a balanced $U(C_n,T_{m_1}, T_{m_2},l)$ is given in Figure \ref{Uni2}.

A directed cycle $dC_n$ is a graph with vertex set $V=\{v_1,v_2,\hdots,v_n\}$ and edge set $E=\{(v_i,v_{i+1})\}$ for $i=1,2,\hdots,n-1$, and $(v_{n},v_1)\in E.$ A graph whose all the blocks are directed cycles is called cactoid graph. We consider a modified version of cactoid graph in which edges can have arbitrary directions and signs, see Figure \ref{balm}.  Adding all the possible arcs (directed edges) between any non adjacent vertices of the  cycle $dC_n$ $(n>3)$ we get a mixed complete graph $mK_n$, see figure \ref{balam} \cite{zhou2017inverse}.
A mix star block graph $G$ is a graph in which complete mixed graph are connected by one cut vertex. An example of mix star block graph is shown in figure \ref{unbalm}.

Any square matrix $A=(a_{ij})$  can be represented by a weighted digraph, $wdG$, in which an edge from vertex $i$ to vertex $j$ has weight equal to $a_{ij}.$ A cycle cover $L$ of $wdG$ is a collection of vertex-disjoint directed cycles that covers all the vertices. The weight, $w(L)$ of cycle cover $L$ is product of the weights of the edges in each directed cycle. Then,

\begin{equation}
\det(A)=(-1)^n\sum_{L}(-1)^{c(L)}w(L),
\end{equation}
\begin{equation}\label{defper}
\text{per}(A)=\sum_{L}w(L),
\end{equation}

where, $c(L)$ is the number of cycles in $L$, and the summation is over all cycle covers. 
 
Determinant of $K_n$ is equal to $(-1)^{n-1} (n-1)$ and permanent of $K_n$ is given by $$\text {per} (K_n)=n!\sum_{i=0}^{n}\frac{(-1)^i}{i!}.$$

In this paper we use $\mathcal{B}$-partitions to calculate the determinant and permanent of above graphs. Paper is organized is as follows: in section \ref{pr}, we give some preliminary results on permanent and determinant of weighted signed graph, and in particular, signed block graph. In section \ref{bppd}, we show how the $\mathcal{B}$-partitions are used to calculated determinant and permanent of block graphs. 
In subsection \ref{bncc}, we calculate the determinant of a block graph having negative vertex disjoint cliques. In the section \ref{unicycle}, we find the determinant and permanent of signed unicyclic graphs. In section \ref{mcmsb}, first, we find eigenvalues of the mixed complete graph and negative mixed complete graph. Thus, we give their determinant expressions. Then we calculate the determinant of mixed star block graph as well as the determinant of negative mixed star block graph.

%In this paper, we first give some preliminary results on permanent and determinant of weighted signed graph, in particular, signed block graph. We then consider conditions for $\mathcal{B}$-partitions of weighted block graphs motivated by \cite{bapat2014adjacency, singh2017characteristic}, and give a general expression for determinant, permanent of weighted block graph. Then we give a formula for permanent of block graph. Finally, we give a formula for the determinant of signed block graph in which each block can have negative disjoint cliques of same sizes, respectively and edges connected to cut-vertices are positive. 

\section{Preliminary results}\label{pr} 
We give some preliminary results on determinant and permanent of signed graphs depending upon whether they are balanced or not.
\begin{theorem} \label{tm}
In a weighted signed block graph $G$ if all the triangles are balanced then, $G$ and $|G|$ have the same determinant. 
\end{theorem} 
 \begin{proof}
From the definition, each block of $G$ is a complete graph. As all the triangle are balanced, every block is a balanced graph \cite{harary1953notion}. Which implies all the cycles in all the blocks of $G$ are balanced. There can not be any common cycle between any two blocks thus, all the cycles of $G$ are balanced hence, $G$ is balanced. From Theorem \ref{bd} $G$ and $|G|$ have same eigenvalues, hence $G$ and $|G|$ have same determinant. 
 \end{proof}

\begin{theorem} \label{pt}
If a weighted signed graph $G$ is balanced then, $G$ and $|G|$ have same permanent. 
\end{theorem} 
\begin{proof}
From \cite{harary1953notion}, a balanced graph can be partitioned into two vertex sets such that all the edges between vertices of same sets are positive while all the edges between vertices of different sets are negative. Let $X, Y$ are two such sets for balanced graph $G$. Let $S$ be the diagonal matrix, whose diagonal elements corresponding to vertices in $X$ are 1 while elements corresponding to vertices in $Y$ are $-1$. Then, $|A|=SAS$. Hence, per$(|A|)$=per($A$)$(\pm 1)^{2}$=per$(A)$. 

\end{proof} 
 
\begin{theorem}
Let $G$ be a weighted signed block graph;  if all the triangles in $G$ are balanced then, $G$ and $|G|$ have same permanent. 
%Let $G$ be a signed block graph; $G$ and $|G|$ have same permanent if and only if all triangles in $G$ are balanced.
\end{theorem} 
\begin{proof}
From the proof of theorem \ref{tm}, if all triangles in signed block graph $G$ are balanced then $G$ is balanced. Now, theorem directly follows from  Theorem \ref{pt}.
\end{proof}

\section{$\mathcal{B}$-partitions in block graphs, determinant and permanent} \label{bppd}
We first state a theorem for determinant of simple block graphs\cite{bapat2014adjacency}. We see that, in theorem, conditions on $k$-tuple $(\alpha_1, \alpha_2,\hdots,\alpha_k )$ can induce  $\mathcal{B}$-partitions and vice and versa. Then, we can write the determinant and permanent of a matrix in using these conditions. The theorem is as follows
\begin{theorem}\cite{bapat2014adjacency}\label{bt}
    Let $G$ be a block graph with $n$ vertices and having all the edges of weight 1. Let $B_1, B_2,\hdots, B_k$ be its blocks. Let $A$ be the adjacency matrix of $G$. Then \begin{equation}\label{be}
    \det (A)=(-1)^{n-k}\sum\prod_{i=1}^{k}(\alpha_i-1),
    \end{equation} 
    where, the summation is over all $k$-tuples $(\alpha_1, \alpha_2, \hdots,\alpha_k)$ of non negative integers satisfying the    
    following conditions:
    \begin{enumerate}
    \item $\sum_{i=1}^{k} \alpha_i=n$;
    \item for any nonempty set $S\subseteq \left\{1,2,\hdots,k \right\}$ $$\sum_{i\in S}\alpha_i\le|V(G_S)|,$$ where $G_S$ denote the subgraph of $G$ induced by the blocks $B_i$, $i\in S$.
    \end{enumerate} 
\end{theorem}
 Note that, with the same conditions on $k$-tuples $(\alpha_1,\alpha_2,\hdots,\alpha_k)$ the equation (\ref{be}) of theorem \ref{bt} can be written as 
 \begin{equation}
 \det(A)=\sum \prod_{i=1}^{k}\det(K_{\alpha_i})
 \end{equation}
 assuming that $\det(K_0)=1$.

Let $wG$ be a weighted digraph having no loops on cut-vertices. In \cite{singh2017characteristic}(corollary 5.1) a combinatorial expression for determinant and permanent of $wG$ given in terms of determinant and permanent of subdigraphs of blocks respectively. Statement for determinant, permanent is as follows. 

\begin{lemma} \label{nocut}
Let $wG$ be a weighted digraph having no loops on its cut-vertices. Let $B_1, B_2,\hdots, B_k$ are blocks in it. Then, the determinant, permanent of $wG$ is given by $$\sum \prod_{i=1}^{k}\det (\hat B_i), \ \sum \prod_{i=1}^{k}\emph{per} \ (\hat B_i),$$ respectively,  where, if $\hat{B_i}$ is a null graph then $\det(\hat{B_i})=1,\ \emph{per}\ (\hat{B_i})=1$. And the summation is over all possible $k$-combination of induced subgraphs $\hat{B_1}, \hat{B_2},\hdots, \hat{B_k}$ such that for $i,j=1,2,\hdots,k,$
\begin{enumerate}
\item $\hat{B_i}\subseteq B_i.$
\item $\bigcup_{i=1}^{k} V(\hat{B_i})=V(wG).$
\item $V(\hat{B_i})\bigcap V(\hat{B_j})=\phi,$ for $i\ne j.$ \end{enumerate}

\end{lemma}

Thus, the summation is over all $k$-combinations $\hat{B_1}, \hat{B_2},\hdots, \hat{B_k}$ of induced subgraphs which partition $wG$. These partitions are called as $\mathcal{B}$-partitions, and corresponding terms $\prod_{i=1}^{k}\det (\hat{B_i})$, $\prod_{i=1}^{k}\text{per} (\hat{B_i})$ are called $\det$-summands, $\text{per}$-summands , respectively.  We will now prove that each $k$-tuple $(\alpha_1, \alpha_2,\hdots,\alpha_k)$ in theorem \ref{bt} produces a unique $\mathcal{B}$-partition of any weighted graph and vice versa.

\begin{lemma}\label{tw}
Let $G$ be a graph with $n$ vertices and $k$ blocks. Let $B_1, B_2,\hdots, B_k$ be its blocks having $b_1,b_2,\hdots,b_k$, number of vertices, respectively. Then, each $\mathcal{B}$-partition produce a unique $k$-tuples $(\alpha_1, \alpha_2, \hdots,\alpha_k)$ of non negative integers satisfying the    
    following conditions:
    \begin{enumerate}
    \item $\sum_{i=1}^{k} \alpha_i=n$;
    \item for any nonempty set $S\subseteq \left\{1,2,\hdots,k \right\}$ $$\sum_{i\in S}\alpha_i\le|V(G_S)|,$$ where $G_S$ denote the subgraph of $G$ induced by the blocks $B_i$, $i\in S$.
    \end{enumerate} 
\end{lemma}
\begin{proof}
From the Lemma \ref{nocut} determinant and permanent of $G$ is equal to 
    $$\sum \prod_{i=1}^{k}\det(\hat{B_i}), \ \sum \prod_{i=1}^{k}\text{per} (\hat{B_i}), $$ respectively,  
    where, $\hat{B_i}$ is subgraph of $B_i$ and summation is over all $\mathcal{B}$-partition of $G$.
    
    $\{\hat{B_1}, \hat{B_2}, \hdots,\hat{B_k}\}$ are vertex disjoint induced subgraphs which create a $\mathcal{B}$-partition of $G$, thus, $\sum_{i=1}^{k}|V(\hat{B_i})|=n$. 
    
    Also, for any nonempty set $S\subset \left\{1,2,\hdots,k \right\}$,
    $$\sum_{i\in S}|V(\hat{B_i})|\le|V(G_S)|,$$ 
    where, $G_S$ denote the subgraph of $G$ induced by the blocks $B_i$, $i\in S$. Assume, the number of vertices in a given $\mathcal{B}$-partition $\hat{B_1},\hat{B_2},\hdots,\hat{B_k} $ be $\alpha_1,\alpha_2,\hdots,\alpha_k$. Thus, $k$-tuples  $(\alpha_1, \alpha_2,\hdots,\alpha_k)$ resulted from $\mathcal{B}$-partitions of $G$ satisfy both the conditions of theorem.     
    
    Conversely, consider a $k$-tuple $(\alpha_1, \alpha_2, \hdots, \alpha_k)$ satisfying both the condition of theorem. We will prove by induction that each such $k$-tuple corresponds to a unique $\mathcal{B}$-partition of $G.$ 
    
    If $G$ has only one block $B_1$ of order $b_1$, then the only possible choice for $1$-tuple is $\alpha_1 = b_1$. Clearly, $\alpha_1$ corresponds to a $\mathcal{B}$-partition which consists of $B_1$ only. Let $G$ has two blocks $B_1$, and $B_2$ of order $b_1$, and $b_2$, respectively, and a cut-vertex $v$.  The possible $2$-tuples are ($\alpha_1=b_1, \ \alpha_2=b_2-1$), and ($\alpha_1=b_1-1, \ \alpha_2=b_2$). Both the $2$-tuple induce possible two $\mathcal{B}$-partitions in $G.$ One $\mathcal{B}$-partition consists of induced subgraphs $B_1, B_2\setminus v$. Another $\mathcal{B}$-partition consists of induced subgraphs $B_1\setminus v, B_2$.
    
    Now we discuss the proof for $G$ consisting of three blocks, which will clarify the reasoning for the general case. For the time being let  us denote the graph having $k$ blocks by $G_k$. Let the blocks are $B_1, B_2,\hdots, B_k$ of order $b_1,b_2,\hdots,b_k$, respectively.  Formation of a $G_k$ can be seen as $k$-step process. At any intermediate $i$-th step a block $B_i$ is added to $G_{i-1}$ and then $B_i$ becomes a pendant block for $G_{i}$. In $G_3$, block $B_3$ can occur in two ways.
    \begin{enumerate}
    \item 
    Let $B_3$ be added to a non cut-vertex of $G_2$. Without loss of generality, let $B_3$ get attached to a non-cut-vertex of $B_2$ in $G_2$. In resulting $G_3$, let $v_1$ be the cut-vertex in $B_1, B_2$, and $v_2$ be the cut-vertex in $B_2, B_3$. Choices for $3$-tuple $(\alpha_1, \alpha_2, \alpha_3)$ are following: 
    \begin{enumerate}
    \item $\alpha_1=b_1,\ \alpha_2=b_2-1,\ \alpha_3=b_3-1$;
    \item $\alpha_1=b_1,\ \alpha_2=b_2-2,\ \alpha_3=b_3$;
    \item $\alpha_1=b_1-1,\ \alpha_2=b_2,\ \alpha_3=b_3-1$;
    \item $\alpha_1=b_1-1,\ \alpha_2=b_2-1,\ \alpha_3=b_3$.
    \end{enumerate}
    Note that, in this case, each 2-tuple of $G_2$ give rise to two 3-tuple in $G_3$ where $\alpha_1$ is unchanged. Clearly, all the tuples in $G_3$ can induce the following all possible $\mathcal{B}$-partitions.  
    \begin{enumerate}
    \item $B_1,\ B_2\setminus v_1,\ B_3\setminus v_2$;
    \item $B_1,\ B_2\setminus (v_1,v_2),\ B_3$;
    \item $B_1\setminus v_1,\ B_2,\ B_3\setminus v_2$;
    \item $B_1\setminus v_1,\ B_2\setminus v_2,\ B_3$.
        \end{enumerate}

    \item Let $B_3$ be added to cut-vertex $v$ of $G_2.$ Choices for $3$-tuple $(\alpha_1, \alpha_2, \alpha_3)$ are following: 
    \begin{enumerate}
    \item $\alpha_1=b_1,\ \alpha_2=b_2-1,\ \alpha_3=b_3-1$;
    \item $\alpha_1=b_1-1,\ \alpha_2=b_2,\ \alpha_3=b_3-1$;
    \item $\alpha_1=b_1-1,\ \alpha_2=b_2-1,\ \alpha_3=b_3$.
    \end{enumerate} 
    Here, each 2-tuple of $G_2$ give rise to a 3-tuple of $G_3$ where $\alpha_1, \alpha_2$ are unchanged and $\alpha_3=b_3-1$. Beside these there is one more 3-tuple where $\alpha_1=b_1-1,\ \alpha_2=b_2-1, \alpha_3=b_3$. Clearly, all the tuples in $G_3$ can induce the following all possible $\mathcal{B}$-partitions.
 
 \begin{enumerate}
     \item $B_1,\ B_2\setminus v,\ B_3\setminus v$;
     \item $B_1\setminus v,\ B_2,\ B_3\setminus v$;
     \item $B_1\setminus v,\ B_2\setminus v,\ B_3$.
         \end{enumerate}

  \end{enumerate}
    
    Now, let us assume that all possible $k$-tuples $(\alpha_1,\alpha_2,\hdots,\alpha_k)$ in $G_k$ can induce all possible $\mathcal{B}$-partitions in it. We need to proof that all possible $(k+1)$-tuples $(\alpha_1,\alpha_2,\hdots,\alpha_k, \alpha_{k+1})$ in $G_{k+1}$ can induce its all possible $\mathcal{B}$-partitions in it. In $G_{k+1}$ block $B_{k+1}$ can occur in two ways.
    \begin{enumerate}
    \item Let $B_{k+1}$ be added to non cut-vertex of $G_k.$ 
    Each $k$-tuple $(\alpha_1,\alpha_2,\hdots,\alpha_k)$ of $G_k$ give rise to two $(k+1)$-tuple of $G_{k+1}$ where, $\alpha_1, \alpha_2,\hdots,\alpha_{k-1}$ are unchanged. In one such tuple $\alpha_k$ is also unchanged and $\alpha_{k+1}=b_{k+1}-1.$ In other tuple $\alpha_{k}$ is one less than the value it had earlier and $\alpha_{k+1}=b_{k+1}.$ Thus, $(k+1)$-tuples in $G_{k+1}$ can induce all its $\mathcal{B}$-partitions in $G_{k+1}.$ 
    
    \item Let $B_{k+1}$ be added to a cut-vertex $v$ of $G_k.$ 
     Each $k$-tuple of $G_k$ give rise to one $(k+1)$-tuple of $G_{k+1}$ where $\alpha_{k+1}=b_{k+1}-1.$
    Beside these there are also $(k+1)$-tuples where $\alpha_{k+1}=b_{k+1}$, along with $k$-tuples of $(G_k\setminus v)$. Clearly, all the tuples in $G_{k+1}$ can induce its
     $\mathcal{B}$-partitions. 
    \end{enumerate}
    Hence, there is one to one correspondence between $\mathcal{B}$- partitions and the $k$-tuples $(\alpha_1,\alpha_2,\hdots ,\alpha_3)$. 
    \end{proof}

Now we give a formula for the permanent of balanced signed block graphs. 
\begin{theorem}\label{btp}
    Let $G$ be a balanced signed block graph with $n$ vertices and having all the edges of weight 1. Let $B_1, B_2,\hdots, B_k$ be its blocks. Let $A$ be the adjacency matrix of $G$. Then, \begin{equation}
     \emph{per}(A)=\sum \prod_{i=1}^{k}\alpha_i!\sum_{j=0}^{\alpha_i}\frac{(-1)^j}{j!},
     \end{equation}
   where, the summation is over all $k$-tuples $(\alpha_1, \alpha_2, \hdots,\alpha_k)$ of non negative integers satisfying the    
    following conditions:
    \begin{enumerate}
    \item $\sum_{i=1}^{k} \alpha_i=n$;
    \item for any nonempty set $S\subseteq \left\{1,2,\hdots,k \right\}$ $$\sum_{i\in S}\alpha_i\le|V(G_S)|,$$ where $G_S$ denote the subgraph of $G$ induced by the blocks $B_i$, $i\in S$.
    \end{enumerate} 
\end{theorem}
 
\begin{proof}
The proof directly follows from Lemma \ref{nocut}, \ref{tw}  and the fact that $$\text{per}(K_{\alpha_i})=\alpha_i!\sum_{j=0}^{\alpha_i}\frac{(-1)^j}{j!}.$$
\end{proof}

\subsection{Block graph with negative cliques.}\label{bncc}
First, we give the determinant of a complete graph with negative cliques, $K^{m,r}_n$. Subsequently, the determinant of block graph with negative cliques is given.  

\begin{lemma}\cite{singh2017eigenvalues}(Corollary 3.6) \label{dnk}
Determinant of $A(K^{m,r}_n)$ is given by $$(1-2r)^{m-1}(-1)^{n-mr-1}\Bigg(n(1-2r)+2r\Big(1+m(r-1)\Big)-1\Bigg).$$
\end{lemma}

\begin{theorem}
Let $G$ be a signed block graph of order $n$ having $k$ blocks $B_1, B_2,\hdots, B_k$. Let all the edges connecting cut-vertices are positive. For $i=1,2,\hdots,k$, let $B_i$ has $m_i$ number of vertex-disjoint negative cliques each of size $r_i$, such that $0\le m_ir_i\le (n_i-1).$
Then, 
\begin{equation}
\det(G)=(-1)^{n-k}\sum \prod_{i=1}^{k} (1-2r_i)^{m_i-1}(-1)^{-m_ir_i}\Bigg(\alpha_i(1-2r_i)+2r_i\Big(1+m_i(r_i-1)\Big)-1\Bigg).
\end{equation}
where, the summation is over all $k$-tuples $(\alpha_1, \alpha_2, \hdots,\alpha_k)$ of non negative integers satisfying the
    
    following conditions:
\begin{enumerate}
    \item $\sum_{i=1}^{k} \alpha_i=n$;
    \item for any nonempty set $S\subseteq \left\{1,2,\hdots,k \right\}$ $$\sum_{i\in S}\alpha_i\le|V(G_S)|,$$ where $G_S$ denote the subgraph of $G$ induced by the blocks $B_i$, $i\in S$.
\end{enumerate}
\end{theorem}
\begin{proof}
The result directly follows from Lemma \ref{tw}, \ref{nocut}, and  \ref{dnk}.
\end{proof}

%\subsection{Cactoid graph with arbitrary directions of edges}
%\begin{theorem}
%A connected weighted cactoid graph $G$ having at least two blocks, and arbitrary directions of edges have zero determinant and permanent. 
%\end{theorem}
%\begin{proof}The proof for permanent is similar, thus we give proof for determinant only.
%Let $B_1$ be the any pendant block of $G$. Let a block $B_2$ has $t$ number of cut-vertices, and it share a common cut-vertex $v$ with the block $B_1$.  From Lemma \ref{nocut}, in det-summands corresponding to all the $\mathcal{B}$-partitions of $G$, either the term $\det(B_1\setminus v)\det \Big(B_2\setminus(S)\Big)$ or $\det(B_1)\det \Big(B_2\setminus(v,S)\Big)$ has to be there, where, $S$ is some subset of cut-vertices of $B_2$ other than $v$.  Since, $B_1\setminus v$ and $B_2\setminus(v,S)$ are directed path graphs, thus no cycle cover is possible in them, hence their determinant is zero.  Which implies $G$ is singular.
%\end{proof}
%

\section{Determinant and permanent of signed unicyclic graphs} \label{unicycle}
Let $U$ be a unicyclic graph which contains a signed cycle $C_n$ as a subgraph with vertices $v_1, v_2,\hdots, v_n$. Let the vertex $v_i$ is linked with $m_i$ number of signed trees say $T^i_1, T^i_2,\hdots, T^i_{m_i},$ such that the root vertex of each $T^i_j, j=1,2,\hdots,m_i$ is linked with $v_i$ by an edge. Note that the vertex $v_i$ then becomes a cut-vertex. As trees are acyclic graph, determinant and permanent of any signed tree is equal to determinant and permanent of its underlying tree with positive edges. Let $\left\{T^i_1, T^i_2,\hdots, T^i_{m_i} \right\}$ denote the subgraph of $U$ induced by the trees $T^i_j, j=1,\hdots,{m_i}.$ Let $U\setminus \left\{T^i_1, T^i_2,\hdots, T^i_{m_i} \right\}$ denotes the induced subgraph of $U$ after $\left\{T^i_1, T^i_2,\hdots, T^i_{m_i} \right\}$ is removed from $U,$ and $\left\{T^i_1, T^i_2,\hdots, T^i_{m_i}, v_i \right\}$ denotes the subgraph of $U$ induced by trees $T^i_1, T^i_2,\hdots, T^i_{m_i}$ and vertex $v_i.$ From \cite{singh2017characteristic}, Lemma 2.3 and Corollary 2.4, can be re-written for determinant and permanent, respectively for graphs with no loop on cut-vertices.

\begin{lemma} \label{recurrsubp}
Let $G$ be a digraph with at least one cut-vertex. Let $H$ be a non empty subdigraph of $G$ having cut-vertex $v$, such that $H\setminus v$ is union of connected components. The determinant of $G$,
\begin{equation}
\det(G)=\det(H)\times\det(G\setminus H)+\det(H\setminus v)\times \det\Big(G\setminus(H\setminus v)\Big).
\end{equation}
\end{lemma}

\begin{corollary}\label{recurrper}
Let $G$ be a digraph with at least one cut-vertex. Let $H$ be a non empty subdigraph of $G$ having cut-vertex $v$, such that $H\setminus v$ is union of connected components. The permanent of $G$,
\begin{equation}
\emph{per} (G)=\emph{per}(H)\times \emph{per}(G\setminus H)+\emph{per}(H\setminus v)\times \emph{per}\Big(G\setminus(H\setminus v)\Big).
\end{equation}
\end{corollary}

Applying Lemma \ref{recurrsubp} on $U$ at $v_i$ we get

\begin{eqnarray}\label{U}
\det(U)&=&\det\Big(U\setminus \left\{T^i_1, T^i_2,\hdots, T^i_{m_i} \right\}\Big)\det\Big(\left\{T^i_1, T^i_2,\hdots, T^i_{m_i} \right\}\Big)\nonumber \\
&& + \det\Big(U\setminus \left\{T^i_1, T^i_2,\hdots, T^i_{m_i}, v_i \right\}\Big)\det\Big(\left\{T^i_1, T^i_2,\hdots, T^i_{m_i},v_i \right\}\Big).
\end{eqnarray}

Applying Corollary \ref{recurrper} on $U$ at $v_i$ we get
\begin{eqnarray}\label{P}
\text{per}(U)&=&\text{per}\Big(U\setminus \left\{T^i_1, T^i_2,\hdots, T^i_{m_i} \right\}\Big)\text{per}\Big(\left\{T^i_1, T^i_2,\hdots, T^i_{m_i} \right\}\Big)\nonumber \\
&& + \text{per}\Big(U\setminus \left\{T^i_1, T^i_2,\hdots, T^i_{m_i}, v_i \right\}\Big)\text{per}\Big(\left\{T^i_1, T^i_2,\hdots, T^i_{m_i},v_i \right\}\Big).
\end{eqnarray}

Then we have the following theorems.

\begin{theorem}\label{Thm:1}
Consider a unicyclic signed graph $U(C_n,T_m)$ where a signed tree $T_m$ is linked with the signed cycle $C_n$ by an edge between the root vertex of $T_m$ and a vertex $v$ of $C_n.$ Then, $$\det\Big(U(C_n,T_m)\Big)=\left\{
\begin{array}{ll}
0, & \hbox{if $n$ is even and $T_m$ has no perfect matching} \\
(-1)^{\frac{m}{2}}\Big(-2\delta+2(-1)^{\frac{n}{2}}\Big), & \hbox{if $n$ is even and $T_m$ has a perfect matching} \\
(-1)^{\frac{m+n}{2}}, & \hbox{if $n$ is odd and $\{T_{m},v\}$ has a perfect matching} \\
2\delta (-1)^{\frac{m}{2}}, & \hbox{if $n$ is odd and $T_m$ has a perfect matching} \\
\end{array}
\right.
$$where $\delta=1$ if $C_n$ is balanced, otherwise $\delta=-1.$
\end{theorem}
\begin{proof}
Let the tree $T_m$ be attached to $C_n$ via an edge between the vertices $u_1$ of $T_m$ and $v$ of $C_n$. Applying Lemma \ref{recurrsubp} the determinant of $U(C_n,T_m)$ can be written as \begin{eqnarray}
\det\Big(U(C_n,T_m)\Big)=\det(C_n)\times \det(T_m)+\det(C_n\setminus v)\times \det (\{T_{m},v\})\nonumber \\
=\det(C_n)\times \det(T_m)+\det(P_{n-1})\times \det(\{T_{m},v\}), \label{unm}
\end{eqnarray}
where, $C_n\setminus v$ is the subgraph in which vertex $v$ is removed from $C_n$ and hence it becomes $P_{n-1}$. Also, a signed tree without a perfect matching has determinant zero. From \cite{singh2017eigenvalues}, Corollary 2.3 

Determinant of signed cycle $C_n$, having weight $\delta\in\{-1,1\}$ is given by
$$\det (C_n)=\begin{cases}
2-2\delta &\mbox{if $n$ is even and even multiple of 2}\\
 -2-2\delta &\mbox{if $n$ is even and odd multiple of 2}\\
 2\delta &\mbox{if $n$ is odd}
\end{cases}$$

Now we consider the following cases.
\begin{enumerate}
\item[Case I] $n$ is even and $T_m$ has not perfect matching: As in this case $\det(T_m)=0, \ \det(P_{n-1})=0$. From equation (\ref{unm})  $\det(U(C_n,T_m))=0.$
\item[Case II] $n$ is even and $T_m$ has a perfect matching: Consider $n=2k, m=2k'$, where $k\ge2$ and $k'\ge1$ are positive integers. As $\det(P_{n-1})=0$ from equation (\ref{unm})
$$\det(U(C_n,T_m))=\det(C_n)\times\det(T_m)=(-2\delta+2(-1)^{k})(-1)^{k'},$$
where, for balanced $C_n$, $\delta=1$ and for unbalanced $C_n$, $\delta=-1.$

\item[Case III] $n$ is odd and as $m$ is odd, $T_m$ has not  perfect matching:
In this case $\det(T_m)=0.$ Thus, from equation (\ref{unm}) $$\det(U(C_n,T_m))=\det(P_{n-1})\times\det(\{T_{m},v\}).$$ If $\{T_{m},v\}$ has no perfect matching then $\det(U(C_n,T_m))=0.$ Otherwise $$\det(U(C_n,T_m))=(-1)^{\frac{n-1}{2}}(-1)^{\frac{m+1}{2}}=(-1)^{\frac{n+m}{2}}.$$

\item[Case IV] $n$ is odd and $T_m$ has a perfect matching:
In this case $m+1$ is an odd number so, $\det(\{T_{m},v\})=0.$ Thus, from equation (\ref{unm}) $$\det(U(C_n,T_m))=\det(C_n)\times\det(T_{m}) =2\delta(-1)^{\frac{m}{2}}$$
where for balanced $C_n$, $\delta=1$ and for unbalanced $C_n$, $\delta=-1.$
\end{enumerate}

\end{proof}

%Now we will analyse all possible cases for determinant of $U(n,m)$
\begin{corollary}
Consider a unicyclic signed graph $U(C_n,T_m)$ as in Theorem \ref{Thm:1}. Then, $$\emph{per}\Big(U(C_n,T_m)\Big)=\left\{
 \begin{array}{ll}
  0, & \hbox{if $n$ is even and $T_m$ has no perfect matching} \\
 -2\delta+2, & \hbox{if $n$ is even and $T_m$ has a perfect matching} \\
 1, & \hbox{if $n$ is odd and $\{T_{m},v\}$ has a perfect matching} \\
2\delta, & \hbox{if $n$ is odd and $T_m$ has a perfect matching} \\
 \end{array}
 \right.
$$ where $\delta=1$ if $C_n$ is balanced, otherwise $\delta=-1.$
\end{corollary}
\begin{proof}
Using equation (\ref{defper})
\begin{eqnarray*} \emph{per}  (C_n)=\begin{cases}
 2-2\delta &\mbox{if $n$ is  even}\\
 2\delta &\mbox{if $n$ is odd.}
\end{cases} \end{eqnarray*}
Rest of the steps are similar to Theorem \ref{Thm:1}.
\end{proof}

\begin{theorem}\label{Thm:2}
Let $U(C_n,\left\{T_{m_1}, T_{m_2},\hdots,T_{m_k} \right\})$ denotes a unicyclic graph having a signed cycle $C_n$ and $k$ signed trees $T_{m_1}, T_{m_2},\hdots, T_{m_k}$ and root of each $T_{m_i}, i=1,\hdots,k$ is linked with vertex $v$ of $C_n$ by an edge for all $i.$ Then
\begin{equation*}
\begin{split}
\det\Big(U(C_n,\left\{T_{m_1}, T_{m_2},\hdots,T_{m_k} \right\})\Big)&=\det(C_n)\prod_{i=1}^{k}\det(T_{m_i})\\
&+\det(P_{n-1})\sum_{i=1}^{k}\Big(\det(\left\{T_{m_i},v\right\})\prod_{j=1,j\ne i}^{k}\det (T_{m_j})\Big).
\end{split}
\end{equation*}
\end{theorem}
\begin{proof}
From equation (\ref{U}) observe that \begin{eqnarray}
\det\Big(U(C_n,\left\{T_{m_1}, T_{m_2},\hdots,T_{m_k} \right\})\Big)&=&\det\Big(U(C_n,\left\{T_{m_1}, T_{m_2},\hdots,T_{m_k} \right\})\setminus \left\{T_{m_1}, T_{m_2},\hdots, T_{m_k}\right\}\Big) \\&&\times \det(\left\{T_{m_1}, T_{m_2},\hdots, T_{m_k}\right\})\nonumber \\&&+
\det\Big(U(C_n,\left\{T_{m_1}, T_{m_2},\hdots,T_{m_k} \right\})\setminus \left\{T_{m_1}, T_{m_2},\hdots, T_{m_k}, v\right\}\Big)\nonumber\\ && \times \det(\left\{T_{m_1}, T_{m_2},\hdots, T_{m_k}, v\right\}).\nonumber
\end{eqnarray}
where, $U(C_n,\left\{T_{m_1}, T_{m_2},\hdots,T_{m_k} \right\})\setminus \left\{T_{m_1}, T_{m_2},\hdots, T_{m_k}\right\} =C_n.$ Also, $$\det\Big(\left\{T_{m_1}, T_{m_2},\hdots, T_{m_k}\right\}\Big)=\prod_{i=1}^{k}\det(T_{m_i})$$ since $\left\{T_{m_1}, T_{m_2},\hdots, T_{m_k}\right \}$ is the induced subgraph of the unicyclic graph having $k$ connected components $T_{m_i}, i=1,\hdots,k$. Next, $U(C_n,\left\{T_{m_1}, T_{m_2},\hdots,T_{m_k} \right\})\setminus \left\{T_{m_1}, T_{m_2},\hdots, T_{m_k}, v\right\}=P_{n-1}$. The only thing that is left to know is $\det(\left\{T_{m_1}, T_{m_2},\hdots, T_{m_k}, v\right\})$. Again applying Lemma \ref{recurrsubp} on $\left\{T_{m_1}, T_{m_2},\hdots, T_{m_k}, v\right\}$ at $v$  $$\det\Big(\left\{T_{m_1}, T_{m_2},\hdots, T_{m_k}, v\right\}\Big)=\sum_{i=1}^{k}\Big(\det(\left\{T_{m_i,v}\right\})\prod_{j=1,j\ne i}^{k}\det (T_{m_j})\Big).$$ Thus, the desired result follows.

%So determinant of $U(n,\left\{m_1, m_2,\hdots,m_k \right\})$ is given by
%\begin{eqnarray}
%\det\Big(U(n,\left\{m_1, m_2,\hdots,m_k \right\})\Big)=\det(C_n)\prod_{i=1}^{k}\det(T_{m_i})+\det(P_{n-1})\sum_{i=1}^{k}\Big(\det(\left\{T_{m_i,v}\right\})\prod_{j=1,j\ne i}^{k}\det (T_{m_j})\Big)
%\end{eqnarray}
\end{proof}
\begin{corollary}
Let $U(C_n,\left\{T_{m_1}, T_{m_2},\hdots,T_{m_k} \right\})$ denote a unicyclic graph as considered in Theorem \ref{Thm:2}. Then
\begin{equation*}
\begin{split}
\emph{per}(U(C_n,\left\{T_{m_1}, T_{m_2},\hdots,T_{m_k} \right\}))&=\emph{per}(C_n)\prod_{i=1}^{k}\emph{per}(T_{m_i})\\
&+\emph{per}(P_{n-1})\sum_{i=1}^{k}\Big(\emph{per}(\left\{T_{m_i},v\right\})\prod_{j=1,j\ne i}^{k}\emph{per} (T_{m_j})\Big).
\end{split}
\end{equation*}
\end{corollary}

\begin{theorem}\label{Thm:3}
Let $U(C_n, T_{m_1}, T_{m_2}, l)$ denote a signed unicyclic graph having a signed cycle $C_n$ and two trees $T_{m_1}, T_{m_2}$ are attached by additional edges to two vertices $v_1$ and $v_2$ of $C_n$ respectively at a distance $l.$ Then,
\begin{equation*}
\begin{split}
\det\Big(U(C_n, T_{m_1}, T_{m_2}, l)\Big)&=\det\Big(U(C_n,T_{m_2})\Big)\det(T_{m_1})\\
&+\det(\left\{T_{m_1},v_1 \right\})\det(\left\{T_{m_2}, v_{l+1}\right\})\det(P_{l-1})\det(P_{n-l-1})\\
&+\det(\left\{T_{m_1},v_1 \right\})\det(\left\{T_{m_2}\right\})\det(P_{n-1}).
\end{split}
\end{equation*}
\end{theorem}
\begin{proof}
By equation (\ref{U}) it follows that
\begin{eqnarray}\label{U(n,m_1)}
\det\Big(U(C_n, T_{m_1}, T_{m_2}, l)\Big) &=& \det\Big(U(C_n, T_{m_1}, T_{m_2}, l)\setminus \left\{T_{m_1} \right\}\Big)\det(T_{m_1})
+ \\ && \det\Big(U(C_n, T_{m_1}, T_{m_2}, l)\setminus \left\{T_{m_1}, v_1 \right\}\Big)\det(\left\{T_{m_1}, v_1 \right\}).
\nonumber
\end{eqnarray}

Note that, $\det\Big(U(C_n, T_{m_1}, T_{m_2}, l)\setminus \left\{T_{m_1} \right\}\Big)=\det\Big(U(C_n,T_{m_2})\Big)$ and $\left\{T_{m_1}, v_1 \right\}$ is a tree with $m_1+1$ vertices. The only thing remains to figure out is $\det\Big(U(C_n, T_{m_1}, T_{m_2}, l)\setminus \left\{T_{m_1}, v_1 \right\}\Big)$. Let for the time being denote $U(C_n, T_{m_1}, T_{m_2}, l)$ by $U$. Applying Lemma \ref{recurrsubp} on $U\setminus \left\{T_{m_1}, v_1 \right\}$ at $v_2$ 
\begin{eqnarray*}
 \det\Big(U\setminus \left\{T_{m_1}, v_1 \right\}\Big)=\det(\left\{T_{m_2}, v_{2}\right\})\det \Bigg(\Big(U\setminus \left\{T_{m_1}, v_1 \right\}\Big)\setminus \left\{T_{m_2},v_{2}\right\}\Bigg)\\
 +\det(\left\{T_{m_2}\right\})\det \Bigg(\Big(U\setminus \left\{T_{m_1}, v_1 \right\}\Big)\setminus \left\{T_{m_2}\right\}\Bigg).
\end{eqnarray*}

Further observe that $\Big(U\setminus \left\{T_{m_1}, v_1 \right\}\Big)\setminus \left\{T_{m_2},v_{2}\right\}$ is a disconnected subgraph with two connected components $P_{l-1}$ and $P_{n-(l+1)}$, and hence
$$\det\Bigg(\Big(U\setminus \left\{T_{m_1}, v_1 \right\}\Big)\setminus \left\{T_{m_2},v_{2}\right\}\Bigg)=\det(P_{l-1})\det(P_{n-l-1}),$$
and $\Big(U\setminus \left\{T_{m_1}, v_1 \right\}\Big)\setminus \left\{T_{m_2}\right\}=P_{n-1}.$ Thus the desired result follows.

%Now $\det (U)$ can be written in terms of the determinant of trees, path graph and $\det\Big(U(n,m_2)\Big)$ which we have already derived in previous sections.
\end{proof}

\begin{corollary}
Let $U(C_n, T_{m_1}, T_{m_2}, l)$ be a signed unicyclic as considered in Theorem \ref{Thm:3}. Then
\begin{equation}
\begin{split}
\emph{per}(U(C_n, T_{m_1}, T_{m_2}, l))&=\emph{per}\Big(U(C_n,T_{m_2})\Big)\emph{per}(T_{m_1})\\
&+\emph{per}(\left\{T_{m_1},v_1 \right\})\emph{per}(\left\{T_{m_2}, v_{2}\right\})\emph{per}(P_{l-1})\emph{per}(P_{n-l-1})\\
&+\emph{per}(\left\{T_{m_1},v_1 \right\})\emph{per}(\left\{T_{m_2}\right\})\emph{per}(P_{n-1}).
\end{split}
\nonumber
\end{equation}
\end{corollary}

\section{Mixed complete graph, mixed star block graph.}\label{mcmsb}
The adjacency matrix $A(mK_n),$ of mix complete graph $mK_n$ can be written as:
$$A(mK_n)=J_n − I_n - Q_n,$$
where, $J_n$ is all one matrix , $I_n$ is an identity matrix, and $Q_n$ is the full-cycle permutation matrix of order $n$. Thus, the $(i,i+1)$-element of $Q_n$ is 1, $i=1,2,\hdots,n-1,$ the $(n,1)$-element of $Q_n$ is 1, and the remaining elements of $Q_n$ are zero \cite{bapat2010graphs}.  

The eigenvalues of $Q_n$ are $w^i(0 \le i \le  n−1)$, and the corresponding
eigenvectors are $$v_i = [1,w^i,w^{2i}, \hdots  ,w^{(n−1)i}]^T$$ for $ 0 \le i \le n − 1,$ where, $w$ is an $n$-th primitive root of 1. The eigenvectors are orthogonal to each other, i.e.
$v^T_i v_j = 0$ for $0 \le i, j \le n − 1$. Note that $v_0$ is all one column vector. Then the eigenvalues of 
$A(mK_n)$ are $\lambda_0 = n-2$ and $\lambda_i = -1-w^i  (1 \le i \le n − 1)$. 

\begin{lemma}\label{lm1}
$$\prod_{i=1}^{n-1}(-1-w^i)=\begin{cases}
0 & \mbox{if $n$ is even}
\\1 & \mbox{if $n$ is odd}
\end{cases}$$
\end{lemma}
\begin{proof}
As $$x^n-1=(x-1)\prod_{i=1}^{n-1}(x-w^i),$$
$$\implies\sum_{i=1}^{n}x^{n-i}=\prod_{i=1}^{n-1}(x-w^i).$$ Hence, the result follows. 
\end{proof}

\begin{theorem}
Determinant of  $A(mK_n)$ is  given by
\begin{equation}
\det \Big(A(mK_n) \Big)=\begin{cases}
0 & \mbox{if $n$ is even}
\\(n-2) & \mbox{if $n$ is odd}
\end{cases}
\end{equation}
\end{theorem}
\begin{proof}
As the eigenvalues of 
$A(mK_n)$ are $\lambda_0 = n-2$ and $\lambda_i = -1-w^i$ for  $(1 \le i \le n − 1)$. 
$$\det \Big(A(mK_n) \Big)=(n-2)\prod_{i=1}^{n-1}(-1-w^i).$$
Now, proof directly follows from Lemma \ref{lm1}.
\end{proof}

\subsection{Mixed star block graph}
A mixed block graph is a strongly
connected directed graph whose blocks are mixed complete graphs.   A mixed block graph having maximum one cut vertex is called mixed star block graph, see figure \ref{unbal}. In other words, a mixed star block graph is obtained from a star cactoid graph after adding all possible directed edges between any two non adjacent vertices in each block. As a star cactoid graph cannot have cycle cover it is evident that it is singular. Let $mK_n\setminus v_i$ denotes a induced subgraph resulting after vertex $v_i$ is removed from $mK_n.$  

\begin{lemma}\label{detminus}
The determinant of $mK_n\setminus v_i (i=1,2,\hdots n)$ is given by $$(-1)^n\Big(\lfloor\frac{n-2}{2}\rfloor\Big).$$
\end{lemma}
\begin{proof}
 Without loss of generality let us remove the first vertex $v_1$ of $mK_n$. Adjacency matrix of  $mK_n\setminus v_1$ can be written as 
$$ A\Big(mK_n\setminus v_1\Big)=
\begin{bmatrix}0& 1 & 1& \dots & 1 \\ 0 & 0 & 1& \ddots & 1\\ 1 & 0 & 0 & \ddots & 1 \\ \vdots & \ddots & \ddots & \ddots & 1 \\ 1 & 1 & 1 & \dots & 0 \end{bmatrix}.$$
 In other words, $A\Big(mK_n\setminus v_1\Big)$ is a square matrix of size $n-1$ whose diagonal and sub-diagonal elements are zero and rest of the elements are 1. Let $R_i$ denotes the $i$-th row of $A\Big(mK_n\setminus v_1\Big)$. In order to calculate its determinant let us first make following elementary row operations.  
\begin{enumerate}
\item $R_i=R_i-R_{i+1} $ for $i=1,2,\hdots (n-2).$
\item Add all the resulting $n-2$ rows in 1. to $(n-1)$-th row. 
\end{enumerate}
These elementary row operations produce following matrix

$$ \begin{bmatrix}0& 1 & 0& \dots & 0 \\ -1 & 0 & 1& \ddots & 0\\ 0 & -1 & 0 & \ddots & 0 \\ \vdots & \ddots & \ddots & \ddots & \ddots \\ 0 & 1 & 1 & \dots & 1 \end{bmatrix}.$$

\begin{figure}
\begin{center}
\begin{tikzpicture}

\tikzset{vertex/.style = {shape=circle,draw,minimum size=3em}}
\tikzset{edge/.style = {->,> = latex'}}
% vertices
\node[vertex] (1) at  (0,0) {$v_1$};
%\node[vertex] (b) at  (4,3) {b};
%\node[vertex] (t) at  (8,0) {t};
%\node[vertex] (d) at  (4,-3) {$t$};
\node[vertex] (2) at (2,0) {$v_2$};
\node[vertex] (3) at (4,0) {$v_3$};
\node[vertex] (n-3) at (8,0) {$v_{n-3}$};
\node[vertex] (n-2) at (10,0) {$v_{n-2}$};
\node[vertex] (n-1) at (12,0) {$v_{n-1}$};

\node at (1, 0.6) {1};
\node at (3, 0.6) {1};
\node at (9, 0.6) {1};
\node at (11, 0.6) {1};
\node at (5, 0.6) {1};
\node at (7, 0.6) {1};
\node at (5, -0.6) {-1};
\node at (7, -0.6) {-1};

\node at (1, -0.6) {-1};
\node at (3, -0.6) {-1};
\node at (9, -0.6) {-1};
\node at (11, -0.6) {1};
\node at (12.3, 0.9) {1};

\node at (10, -1.29) {1};
\node at (8, -2.29) {1};
\node at (8, -1.70) {1};
\node at (7, -3.2) {1};

\draw[edge] (1)  to[bend left] (2);
\draw[edge]  (2)   to[bend left] (1) ;

\draw[edge] (2) to[bend left] (3);
\draw[edge] (3) to[bend left] (2);

\node [shape=circle,minimum size=3em] (a3) at (6,0) {};
\draw[edge] (3) to[bend left] (a3);
\draw[edge] (a3) to[bend left] (3);
\draw[edge] (a3) to[bend left] (n-3);
\draw[edge] (n-3) to[bend left] (a3);

\draw[edge] (n-3) to[bend left] (n-2);
\draw[edge] (n-2) to[bend left] (n-3);
\draw[edge] (n-2) to[bend left] (n-1);
\draw[edge] (n-1) to[bend left] (n-2);
\draw[edge] (n-1) to[bend left=60] (2);
\draw[edge] (n-1) to[bend left=50] (a3);
\draw[edge] (n-1) to[bend left=50] (3);
\draw[edge] (n-1) to[bend left=45] (n-3);
\draw[edge] (n-1) to[loop above] (n-1);
%\draw[edge] (n-1)[loop above] (n-1)
\path (3) to node {\dots \dots \dots} (n-3);
\end{tikzpicture}
\caption{Digraph of matrix $A\Big(mK_n\setminus v_1\Big)$ after elementary operations.} 
\label{example}
\end{center}
\end{figure}
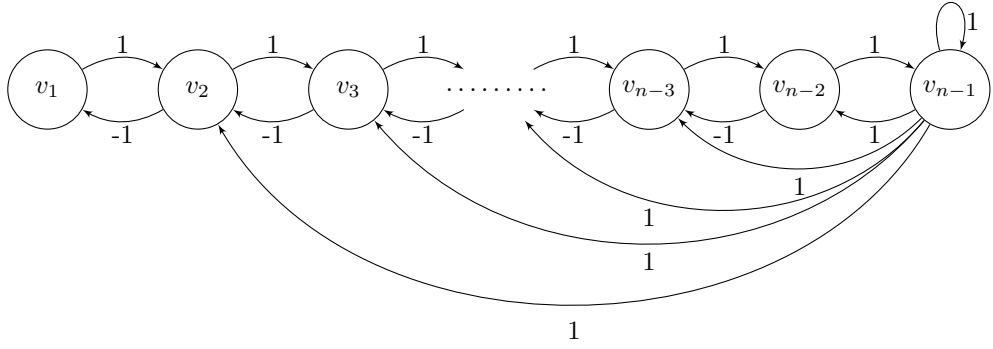

Digraph corresponding to above matrix is shown in figure \ref{example}. Using cycle covers of digraph we calculate the determinant as follows.

\begin{enumerate}
\item $n$ is odd: In this case cycle covers are following. For $i=1,2,\hdots, \frac{n-3}{2}$, in a cycle cover there are directed 2-cycles, each having weight -1, on vertices $\{v_{2j-1},v_{2j}\}$ $(j=1,2\hdots,i)$, and a directed $(n-1-2i)$-cycle of weight 1 on vertices $\{v_{n-1},v_{2i+1},v_{2i+2},\hdots,v_{n-1}\}$. Hence,
\begin{equation}\label{odd}
\begin{split}
det\Bigg(A\Big(mK_n\setminus v_1\Big)\Bigg)=(-1)^{n-1}&\sum_{i=1}^{\frac{n-3}{2}}(-1)^{i+1}\times(-1)^i\times 1 \\
=\frac{3-n}{2}.
\end{split}
\end{equation}

\item $n$ is even: In this case cycle covers are following. For $i=1,2,\hdots, \frac{n-4}{2}$, in a cycle cover there are directed 2-cycles, each having weight -1, on vertices $\{v_{2j-1},v_{2j}\}$ $(j=1,2\hdots,i)$, and a directed $(n-1-2i)$-cycle of weight 1 on vertices $\{v_{n-1},v_{2i+1},v_{2i+2},\hdots,v_{n-1}\}$. Other than these there is one more cycle cover having loop at vertex  $v_{n-1}$, and $\frac{n-2}{2}$ directed 2-cycles on $\{v_{2i-1},v_{2i}\} (i=1,2,\hdots,\frac{n-2}{2} )$ each of weight -1. Hence,

\begin{equation}\label{even}
\begin{split}
det\Bigg(A\Big(mK_n\setminus v_1\Big)\Bigg)=(-1)^{n-1}\Bigg(\sum_{i=1}^{\frac{n-4}{2}}&(-1)^{i+1}\times(-1)^i\times 1+(-1)^{1+\frac{n-2}{2}}\times(-1)^{\frac{n-2}{2}}\times 1\Bigg)\\
=\frac{n-2}{2}.
\end{split}
\end{equation}

\end{enumerate}

Combining  equation (\ref{odd}), (\ref{even}) result follows.
\end{proof}

%\begin{lemma}\cite{singh2017characteristic}(corollary 5.1) \label{nocut}
%Let $G$ be a digraph having no loops on its cut-vertices. Let $B_1, B_2,\hdots, B_k$ are blocks in it. Then, the determinant of $G$ is given by $$\sum \prod_{i=1}^{k}\det \hat B_i,$$ where summation is over all possible $k$-combination of $\hat{B_1}, \hat{B_2},\hdots, \hat{B_k}$ such that for $i,j=1,2,\hdots,k,$
%\begin{enumerate}
%\item $\hat{B_i}\subseteq B_i.$
%\item $\bigcup_{i=1}^{k} V(\hat{B_i})=V(G).$
%\item $V(\hat{B_i})\bigcap V(\hat{B_j})=\phi,$ for $i\ne j.$
%\end{enumerate}
%
%\end{lemma}

\begin{theorem}\label{mbg}
Let $mG$ be mixed star block graph having $k$ blocks $B_1, B_2,\hdots, B_k$ of order $n_1,n_2,\hdots,n_k$, respectively, then $$\det(mG)= \sum\det(mK_{n_i})\prod_{j=1,j\ne i}^{k} (-1)^{n_j}\Big(\lfloor\frac{n_j-2}{2}\rfloor\Big),$$ where summation is over all $i$ such that $n_i$ is odd.
\end{theorem}
\begin{proof} Let $v$ be the cut-vertex of $mG$.
From Lemma \ref{detminus} and \ref{nocut} 

\begin{equation}
\begin{split}
\det(mG)= \sum_{i=1}^{k}&\det(mK_{n_i})\prod_{j=1,j\ne i}^{k}\det\Big(mK_{n_i}\setminus v\Big)\\
= \sum_{i=1}^{k}&\det(mK_{n_i})\prod_{j=1,j\ne i}^{k} (-1)^{n_j}\Big(\lfloor\frac{n_j-2}{2}\rfloor\Big).
\end{split}
\end{equation}
from Lemma \ref{lm1}, for even $n_i$, $\det(mK_{n_i})=0.$ Hence,
$$\det(mG)=\sum\det(mK_{n_i})\prod_{j=1,j\ne i}^{k} (-1)^{n_j}\Big(\lfloor\frac{n_j-2}{2}\rfloor\Big),$$ where, summation is over all $i$ such that $n_i$ is odd.

\end{proof}

\subsection{Negative mix complete graph}
A negative directed cycle $dC_n$ is cycle graph whose each directed edge is negative that is each of its edges have weight $-1$. A negative mixed complete graph $\overline{m}K_n$ is obtained from a negative directed cycle $dC_n$ of length $n > 3$ by adding all the possible positive arcs between any non-adjacent vertices of the underlying cycle $C_n$. Adjacency matrix $A(\overline{m}K_n)$ can be written as:
$$A(\overline{m}K_n)=J_n − I_n -2Q_n-Q^{n-1},$$
where, $J_n$ is all one matrix , $I_n$ is an identity matrix, and $Q_n$ is the full-cycle permutation matrix of order $n$. Then the eigenvalues of 
$A(\overline{m}K_n)$ are $\lambda_0 = n-4$ and $\lambda_i = -1-2w^i-w^{i(n-1)}  (1 \le i \le n − 1)$, where $w=e^{\frac{2\pi \iota}{n}}.$
\begin{lemma}
Determinant of  $A(\overline{m}K_n)$ is  given by
\begin{equation}
\det \Big(A(\overline{m}K_n) \Big)=\begin{cases}
2(n-4)\prod_{i=1}^{\frac{(n-2)}{2}} \Bigg(2+8 \cos^2 \Big(\frac{2\pi i}{n}\Big)+6 \cos \Big(\frac{2\pi i}{n}\Big)\Bigg), & \mbox{if $n$ is even}
\\(n-4)\prod_{i=1}^{\frac{(n-1)}{2}}\Bigg(2+8 \cos^2 \Big(\frac{2\pi i}{n}\Big)+6 \cos \Big(\frac{2\pi i}{n}\Big)\Bigg), & \mbox{if $n$ is odd.}
\end{cases}
\end{equation}
\end{lemma}
\begin{proof}
For $i=1,2,\hdots,(n-1), w^i=\cos \Big(\frac{2\pi i}{n}\Big)+\iota \sin \Big(\frac{2\pi i}{n}\Big)$, and 
\begin{equation*} 
\begin{split}
\lambda_i = -1-2w^i-w^{i(n-1)} 
& \\=-1-2w^i-w^{-i}
& \\=-1-3\cos \Big(\frac{2\pi i}{n}\Big)-\iota\sin\Big(\frac{2\pi i}{n}\Big).
\end{split}
\end{equation*}
Now, $3\cos \Big(\frac{2\pi (n-i)}{n}\Big)-\iota\sin\Big(\frac{2\pi (n-i)}{n}\Big)=3\cos \Big(\frac{2\pi i}{n}\Big)+\iota\sin\Big(\frac{2\pi i}{n}\Big),$ if $n$ is even then, $\lambda_{n/2}=2.$ Following are the determinant expressions for $A(\overline{m}K_n)$ . 

\begin{enumerate}
\item $n$ is odd: \begin{equation*}
\begin{split}
\det\Big(A(\overline{m}K_n)\Big)=(n-4)\prod_{i=1}^{\frac{(n-1)}{2}} \Bigg( \Big( -1-3\cos \Big(\frac{2\pi i}{n}\Big)\Big)^2+  \sin^2\Big(\frac{2\pi i}{n}\Big)\Bigg)&\\= (n-4)\prod_{i=1}^{\frac{(n-1)}{2}}\Bigg(2+8 \cos^2 \Big(\frac{2\pi i}{n}\Big)+6 \cos \Big(\frac{2\pi i}{n}\Big)\Bigg).
\end{split}
\end{equation*}

\item $n$ is even: \begin{equation*}\begin{split}
\det\Big(A(\overline{m}K_n)\Big)=2(n-4)\prod_{i=1}^{\frac{(n-2)}{2}} \Bigg(2+8 \cos^2 \Big(\frac{2\pi i}{n}\Big)+6 \cos \Big(\frac{2\pi i}{n}\Big)\Bigg).
\end{split}
\end{equation*}

\end{enumerate}
\end{proof}

\subsection{Determinant of negative mixed star block graph}
A negative mixed block graph is a strongly
connected directed graph whose blocks are negative mixed complete graphs. A negative mixed block graph having maximum one cut vertex is called negative mixed star block graph. Let $\overline{m}K_n\setminus v_i$ denotes a induced subgraph resulting after vertex $v_i$ is removed from $\overline{m}K_n.$  

\begin{lemma}\label{detminus1}
The determinant of $\overline{m}K_n\setminus v_i (i=1,2,\hdots n)$ is given by $$\Bigg(1+\frac{1}{g_{n-1}}\Big(\sum_{i\le j}2^{j-i}g_{i-1}h_{j+1}+\sum_{j<i}g_{j-1}h_{i+1}\Big)\Bigg)g_{n-1},$$
where,  
$$g_i=r_{1}s_1^i+r_{2}s_2^i, \ \ \text{for} \ \ i=2, 3\dots,n-1,$$  $$h_i=r_{h1}s_1^{n-1-i}+r_{h2}s_2^{n-1-i},\ \  \text{for} \ \ i=n-2,\dots,1,$$
   $$r_1=\frac{1}{2} +\frac{\iota }{2\sqrt{7}},\ \ r_2=\frac{1}{2} -\frac{\iota }{2\sqrt{7}},\ \ r_{h1}=\frac{-1}{2}+\frac{3\iota}{2\sqrt(7)},\ \ \ r_{h2}=\frac{-1}{2}-\frac{3\iota}{2\sqrt(7)},\ \text{and}$$

$$ s_1=\frac{-1}{2}+ \frac{\iota \sqrt{7}}{2},\ \ s_2=\frac{-1}{2}- \frac{\iota \sqrt{7}}{2}.$$
\end{lemma}
\begin{proof}
 Without loss of generality let us remove the first vertex $v_1$ of $\overline{m}K_n$. Adjacency matrix of  $\overline{m}K_n\setminus v_1$ can be written as 
$$ A\Big(\overline{m}K_n\setminus v_1\Big)=
\begin{bmatrix}0& -1 & 1& \dots & 1 \\ 0 & 0 & -1& \ddots & 1\\ 1 & 0 & 0 & \ddots & 1 \\ \vdots & \ddots & \ddots & \ddots & -1\\ 1 & 1 & 1 & \dots & 0 \end{bmatrix}.$$

Let $m=n-1.$ We can write, $A\Big(\overline{m}K_n\setminus v_1\Big)=uu^T+T,$ where, $u$ is a $m\times 1$ column vector having all entries equal to 1. And, $T$ is the tridiagonal matrix of order $m$, having diagonal, subdiagonal entries equal to $-1$ and superdiagonal entries equal to $-2$.  

The proof of this lemma by the first author can also found in \cite{264264}, an alternate expression of the proof can be found in \cite{264167}. From matrix determinant lemma \cite{ding2007eigenvalues}

$$\det(T+uu^T)=(1+u^TT^{-1}u)\det(T).$$

 From \cite{zhou2017inverse}, we need to solve some recursive expressions, in order to calculate the determinant and inverse of $T$. We solve these recursive expressions using roots of their characteristic equations. For the determinant of $A$, recursive expression is $$f_m=-f_{m-1}-2f_{m-2},\ \ \ \ \ f_0=1,\ f_{-1}=0.$$  Roots of the resulting characteristic equation $x^2+x+2=0,$ are $$ s_1=\frac{-1}{2}+ \frac{\iota \sqrt{7}}{2},\ \ s_2=\frac{-1}{2}- \frac{\iota \sqrt{7}}{2}.$$ Hence,   $$\det(T)=f_m=r_1s_1^m+r_2s_2^m,$$ where, using initial conditions $$r_1=\frac{1}{2} +\frac{\iota }{2\sqrt{7}},\ \ r_2=\frac{1}{2} -\frac{\iota }{2\sqrt{7}}.$$ 
Now, to calculate $T^{-1}$ we need to solve following recursive expressions $$g_i=-g_{i-1}-2g_{i-1}, \ \text{for} \ \ i=2, 3\dots,m,\ \  g_0=1,\ g_1=-1$$    $$ h_i=-h_{i+1}-2h_{i+2}, \ \text{for} \ \ i=m-1,\dots,1, \ \ h_{m+1}=1,\ h_{m}=-1.$$ Similar to $f_n$, solving these recursive expressions we get

 $$g_i=r_{1}s_1^i+r_{2}s_2^i, \ \ \text{for} \ \ i=2, 3\dots,n,$$ and, $$h_i=r_{h1}s_1^{m-i}+r_{h2}s_2^{m-i},\ \  \text{for} \ \ i=m-1,\dots,1,$$
  where, $$r_{h1}=\frac{-1}{2}+\frac{3\iota}{2\sqrt(7)},\ \ \ r_{h2}=\frac{-1}{2}-\frac{3\iota}{2\sqrt(7)}.$$ Entries of $T^{-1}$ are clearly given by $g_i, h_i$ \cite{ding2007eigenvalues}. 

$$T^{-1}_{ij}=\begin{cases}\frac{2^{j-i}g_{i-1}h_{j+1}}{g_m}
 & \mbox{if $i\le j$}
\\ \frac{g_{j-1}h_{i+1}}{g_m} &\mbox{if $j< i$}
\end{cases}.$$

As, $u^TT^{-1}u$ equals to sum of all the entries of $T^{-1}$. Thus,
\begin{equation}
u^TT^{-1}u=\frac{1}{g_m}\Big(\sum_{i\le j}2^{j-i}g_{i-1}h_{j+1}+\sum_{j<i}g_{j-1}h_{i+1}\Big)
\end{equation}
Hence, determinant of $\overline{m}K_n\setminus v_i(i=1,2,\hdots,n)$ is given by
$$\Bigg(1+\frac{1}{g_m}\Big(\sum_{i\le j}2^{j-i}g_{i-1}h_{j+1}+\sum_{j<i}g_{j-1}h_{i+1}\Big)\Bigg)g_{n-1}.$$

\end{proof}

\begin{theorem}
Let $\overline{m}G$ be mixed star negative block graph having $k$ blocks $B_1, B_2,\hdots, B_k$ of order $n_1,n_2,\hdots,n_k$, respectively, then $$\det(\overline{m}G)= \sum_{i=1}^{k}\det(\overline{m}K_{n_i})\prod_{j=1,j\ne i}^{k} D_n.$$
\end{theorem}
\begin{proof} Proceeding as the proof of Theorem \ref{mbg} the result directly follows from Lemma \ref{detminus} and \ref{nocut}. 

\end{proof}

%\section{Conclusion}
%Yet to be written..

\bibliographystyle{plain}
\bibliography{BPA}
\end{document}